\documentclass[runningheads]{llncs}
\usepackage{epsfig,verbatim}

\newtheorem{observation}{Observation}

\usepackage{amssymb,amsmath,graphicx}
\usepackage{xspace}
\usepackage{times}

\newtheorem{claimx}{Claim}

\newcommand{\remove}[1]{}

\def\B{{\cal B}}

\newcommand{\bcmd}{{\sc bcmd}\xspace}
\newcommand{\spt}{{\sc spt}}
\newcommand{\mhspt}[1]{{\sc mh$_#1$spt}\xspace}

\newcommand{\W}[1][XXX]{{\normalfont W[#1]}}
\newcommand{\NP}{{\normalfont NP}}

\newcommand{\Card}[1]{|#1|}
\newcommand{\Yes}{\textsc{Yes}}

\DeclareMathOperator{\dist}{dist}

\renewenvironment{proof}
{{\bf Proof.}}{\hspace*{\fill}$\Box$\par\vspace{2mm}}

\begin{document}

\title{Augmenting Graphs to Minimize the Diameter}

\author{Fabrizio Frati\inst{1} \and Serge Gaspers\inst{2,3}\and \\Joachim Gudmundsson\inst{1,3} \and Luke Mathieson\inst{4}
\institute{University of Sydney, Australia, \email{brillo@it.usyd.edu.au,joachim.gudmundsson@sydney.edu.au}
 \and The University of New South Wales, Australia, \email{sergeg@cse.unsw.edu.au}
 \and NICTA
, Australia
 \and Macquarie University, Australia, \email{luke.mathieson@mq.edu.au}}}

\maketitle

\begin{abstract}
We study the problem of augmenting a weighted graph by inserting edges of bounded total cost while minimizing the diameter of the augmented graph. Our main result is an FPT $4$-approximation algorithm for the problem.
\end{abstract}

\vspace{0.5mm}
\section{Introduction}

We study the problem of minimizing the diameter of a weighted graph by the insertion of edges of bounded total cost. This problem arises in practical applications~\cite{bgp-ianar-12,dz-mdnus-10} such as telecommunications networks, information networks, flight scheduling, protein interactions, and it has also received considerable attention from the graph theory community, see for example~\cite{agr-ddbdg-99,er-ptg-62,g-ddc-03}.

We introduce some terminology. Let $G=(V,E)$ be an undirected weighted graph. Let $[V]^2$ be the set of all possible edges on the vertex set $V$. A {\em non-edge} of $G$ is an element of  $[V]^2 \setminus E$. The \emph{weight} of a path in $G$ is the sum of its edge weights. For any $u,v\in V$, the {\em shortest $u$-$v$ path} in $G$ is the path connecting $u$ and $v$ in $G$ with minimum weight. The weight of this path is said to be the \emph{distance} between $u$ and $v$ in $G$.
Finally, the \emph{diameter} of $G$ is the largest distance between any two vertices in $G$.
The problem we study in this paper is formally defined as follows.

\vspace{3mm}

\begin{tabular}{l p{9.5cm}}
{\sc Problem:} & Bounded Cost Minimum Diameter Edge Addition (\bcmd)\\
{\sc Input:} & An undirected graph $G=(V,E)$, a weight function $w:[V]^2 \rightarrow \mathbb{N}$, a cost function $c:[V]^2 \rightarrow \mathbb{N}^*$, and an integer $B$.\\
{\sc Goal:} & A set $F$ of non-edges with $\sum_{e\in F} c(e)\leq B$ such that the diameter of the graph $G_B=(V,E\cup F)$ with weight function $w$ is minimized. We say that $G_B$ is a \emph{$B$-augmentation} of~$G$.\\
\end{tabular}

\vspace{3mm}

The main result of this paper is a fixed parameter tractable (FPT) $4$-approxi\-ma\-tion algorithm for \bcmd with parameter $B$. FPT approximation algorithms are surveyed by Marx~\cite{m-pcaa-08}. For background on parameterized complexity we refer to \cite{DowneyF99,FlumG06,Niedermeier06} and for background on approximation algorithms to \cite{v-aa-01}.

Several papers in the literature already dealt with the \bcmd problem. However, most of them focused on restricted versions of the problem, namely the one in which all costs and all weights are identical~\cite{cv-atmcd-02,dk-dnbpd-99,ks-bdmcg-07,lms-mcbdb-92}, and the one in which all the edges have unit costs and the weights of the non-edges are all identical \cite{bgp-ianar-12,dz-mdnus-10}.

The \bcmd problem can be seen as a bicriteria optimization problem where the two optimization criteria are: (1) the cost of the edges added to the graph and (2) the diameter of the augmented graph. As is standard in the literature, we say that an algorithm is an $(\alpha,\beta)$-approximation algorithm for the \bcmd problem, with $\alpha, \beta \geq 1$, if it computes a set $F$ of non-edges of $G$ of total cost at most $\alpha\cdot B$ such that the diameter of $G'=(V,E\cup F)$ is at most $\beta\cdot D^B_{opt}$, where $D^B_{opt}$ is the diameter of an optimal $B$-augmentation of $G$.

We survey some known results about the \bcmd problem. Note that all the algorithms discussed below run in polynomial time.

\paragraph{Unit weights and unit costs.} The restriction of \bcmd to unit costs and unit weights was first shown to be \NP-hard in 1987 by Schoone et al.~\cite{sbl-diced-87}; see also the paper by Li et al.~\cite{lms-mcbdb-92}. Bil\`{o} et al.~\cite{bgp-ianar-12} showed that, as a consequence of the results in~\cite{cv-atmcd-02,dk-dnbpd-99,lms-mcbdb-92}, there exists no $(c \log n, \delta<1+1/D^B_{opt})$-approximation algorithm for \bcmd if $D^B_{opt}\geq 2$, unless P=\NP. For the case in which $D^B_{opt}\geq 6$, they proved a stronger lower bound, namely that there exists no $(c \log n, \delta<\frac {5}{3}-\frac {7-(D^B_{opt}+1) \mod 3} {3D^B_{opt}})$-approximation algorithm, unless P=\NP.

Dodis and Khanna~\cite{dk-dnbpd-99} gave an $(O(\log n),2+ 2/D^B_{opt})$-approximation algorithm (see also~\cite{ks-bdmcg-07}). Li et al.~\cite{lms-mcbdb-92} showed a $(1,4+ 2/D^B_{opt})$-approximation algorithm. The analysis of the latter algorithm was later improved by Bil\`{o} et al.~\cite{bgp-ianar-12}, who showed that it gives a $(1,2+2/D^B_{opt})$-approximation. In the same paper they also gave a $(O(\log n), 1)$-approximation algorithm.

\paragraph{Unit costs and restricted weights.}
Some of the results from the unweighted setting have been extended to a restricted version of the weighted case, namely the one in which the edges of $G$ have arbitrary non-negative integer weights, however all the non-edges of $G$ have cost $1$ and uniform weight $\omega \ge 0$.

Bil\`{o} et al.~\cite{bgp-ianar-12} showed how two of their algorithms can be adapted to this restricted weighted case. In fact, they gave a $(1,2+2\omega/D^B_{opt})$-approximation algorithm and a $(2-1/B,2)$-approximation algorithm.  Similar results were obtained by Demaine and Zadimoghaddam in~\cite{dz-mdnus-10}.

Bil\`{o} et al.~\cite{bgp-ianar-12} also showed that, for every $D^B_{opt}\geq 2\omega$ and for some constant $c$, there is no $(c \log n,\delta<2- 3\omega/D^B_{opt})$-approximation algorithm for this restriction of the \bcmd problem, unless P=\NP.

\paragraph{Arbitrary costs and weights.}
To the best of our knowledge, there is only one theory paper that has considered the general \bcmd problem.
In 1999, Dodis and Khanna~\cite{dk-dnbpd-99} presented an $O(n \log D^B_{opt},1)$-approximation algorithm, assuming that all weights are polynomially bounded. Their result is based on a multicommodity flow formulation of the problem.

\paragraph{Our results.} In this paper we study the \bcmd problem with arbitrary integer costs and weights. Our main result is a $(1,4)$-approximation algorithm with running time $O((3^B B^3 + n + \log (Bn)) B n^2)$. We also prove that, considering $B$ as a parameter, it is $\W[2]$-hard to compute a $(1+c/B, 3/2-\epsilon)$-approximation, for any constants $c$ and $\epsilon>0$. Further, we present polynomial-time $((k+1)^2,3)$-, $(k,4)$-, and $(1,3k+2)$-approximation algorithms for the unit-cost restriction of the \bcmd problem.

\section{Shortest Paths with Bounded Cost}

Let $(G=(V,E),w,c,B)$ be an instance of the \bcmd problem and let $K$ denote the complete graph on the vertex set $V$. The edges of $K$ have the same weights and costs as they have in $G$ (observe that an edge $e$ of $K$ is either an edge or a non-edge of $G$). For technical reasons, we add self-loops with weight $0$ and cost $1$ at each vertex of $K$.

For any $0\leq \beta \leq B$, a path in $K$ is said to be a {\em $\beta$-bounded-cost path} if it uses non-edges of $G$ of total cost at most $\beta$. We consider the problem of computing, for {\em every} integer $0\leq \beta \leq B$ and for every two vertices $u,v \in V$, a $\beta$-bounded-cost shortest path connecting $u$ and $v$, if such a path exists. We call this problem the \emph{All-Pairs $B$-Shortest Paths} (APSP$_B$) problem. We will prove the following.

\begin{theorem} \label{th:apsp}
The APSP$_B$ problem can be solved in $O(B n^3 + Bn^2 \log (Bn))$ time using $O(B n^2)$ space.
\end{theorem}

\noindent
In order to prove Theorem~\ref{th:apsp}, we construct a directed graph $H=(U,F)$ as follows. First, consider $G$ as a directed graph, i.e., replace every undirected edge $\{u,v\}$ with two arcs $(u,v)$ and $(v,u)$ with the same weight and cost as the edge $\{u,v\}$. Then, $H=(U,F)$ contains $B+1$ copies of $G$, denoted by $G_0, \ldots, G_B$. For any $0\leq i\leq B$, we denote by $(v,i)$ the copy of vertex $v\in V$ in $G_i=(V_i,E_i)$. The arc set $F$ contains the union of $E'$ and $F'$, where $E' = \bigcup_{0\leq i \leq k} E_i$, and $$F'=\Big \{ \big ( (u, i),(v,i+c(\{u,v\}))\big ) \: : \: 0 \leq i \leq B-c(\{u,v\}), \: \{u,v\} \in [V]^2 \setminus E \Big \}.$$
For each $((u,i),(v,j))\in F'$, the weight and the cost of $((u,i),(v,j))$ are $w(\{u,v\})$ and $c(\{u,v\})=j-i$, respectively.

\begin{observation} \label{obs:vertices-edges}
  The number of vertices in $U$ is $(B+1)n$ and the number of arcs in $F$ is $O(Bn^2)$.
\end{observation}

\noindent
We will use directed graph $H$ to efficiently compute $\beta$-bounded-cost shortest paths in $K$. This is possible due to the following two lemmata.

\begin{lemma}\label{le:correspondence1}
Suppose that $H$ contains a directed path $P_H$ with weight $W$ connecting vertices $(u,i)$ and $(v,j)$, for some $j\geq i$. Then, there exists a $(j-i)$-bounded-cost path $P_K$ in $K$ with weight $W$ connecting $u$ and $v$.
\end{lemma}

\begin{proof}
Consider a directed path $P_H$ in $H$ with weight $W$ connecting vertices $(u,i)$ and $(v,j)$, for some $j\geq i$. We define a path $P_K$ in $K$ as follows. Path $P_K$ has the same number of vertices of $P_H$. Also, for any $1\leq j\leq |P_H|$, if the $j$-th vertex of $P_H$ is a vertex $(w,i)$, then the $j$-th vertex of $P_K$ is vertex $w$. Observe that $P_K$ connects vertices $u$ and $v$. We prove that $P_K$ is a $(j-i)$-bounded-cost path with weight $W$. Every edge of $P_H$ connecting vertices $(x,a)$ and $(y,a)$ corresponds to an edge $(x,y)$ of $K$ that is an edge of $G$ with the same weight. Moreover, every edge of $P_H$ connecting vertices $(x,a)$ and $(y,b)$ with $b>a$ either corresponds to an edge $(x,y)$ of $K$ that is a non-edge of $G$ with cost $b-a$ and with the same weight, or it corresponds to a self-loop in $K$ with the same weight (that is $0$). Hence, $P_K$ uses non-edges of $G$ of total cost at most $j-i$ and total weight $W$. Thus, $P_K$ is a $(j-i)$-bounded-cost path with weight $W$ connecting $u$ and $v$, and the lemma follows.
\end{proof}

\begin{lemma} \label{le:correspondence2}
Suppose that there exists a $\beta$-bounded-cost path $P_K$ in $K$ with weight $W$ connecting vertices $u$ and $v$. Then, there exists a directed path $P_H$ in $H$ with weight $W$ connecting vertices $(u,0)$ and $(v,\beta)$.
\end{lemma}

\begin{proof}
Consider a path $P_K=\langle v_1, v_2, \ldots , v_m \rangle$ in $K$ with weight $W$. Set $(v_1,0)$ to be the first vertex of $P_H$. Suppose that path $P_H$ has been defined until a vertex $(v_h,j)$, corresponding to vertex $v_h$ of $P_K$, for some $1\leq h <m$. If edge $(v_h,v_{h+1})$ of $P_K$ is an edge of $G$, then let $(v_{h+1},j)$ be the vertex corresponding to $v_{h+1}$. If edge $(v_h,v_{h+1})$ of $P_K$ is a non-edge of $G$, then let $(v_{h+1},j+c(\{v_h,v_{h+1}\}))$ be the vertex corresponding to $v_{h+1}$. This defines path $P_H$ up to a vertex $(v,\beta')$. Assuming that $\beta'\leq \beta$, path $P_H$ terminates with a set of edges with weight $0$ connecting $(v,j)$ and, $(v,j+1)$, for every $\beta'\leq j\leq \beta-1$; these edges exist by construction. It remains to prove that $\beta'\leq \beta$ and that $P_H$ has weight $W$. Every edge $(x,y)$ of $P_K$ that is an edge of $G$ corresponds to an edge of $H$ connecting vertices $(x,a)$ and $(y,a)$ with the same weight. Moreover, every edge $(x,y)$ of $P_K$ that is a non-edge of $G$ corresponds to an edge of $H$ connecting vertices $(x,a)$ and $(y,b)$, with $c\{x,y\}=b-a$ and with the same weight. By definition, $P_K$ uses non-edges of $G$ of total cost at most $\beta$. Hence, $\beta'\leq \beta$; also, $P_H$ has weight exactly $W$ and the lemma follows.
\end{proof}

\noindent
We have the following.

\begin{corollary} \label{cor:equivalence}
There is a $\beta$-bounded-cost path connecting vertices $u$ and $v$ in $K$ with weight $W$ if and only if there is a directed path in $H$ connecting vertices $(u,0)$ and $(v,\beta)$ with weight $W$.
\end{corollary}

\begin{proof}
The necessity follows from Lemma~\ref{le:correspondence2}. The sufficiency follows from Lemma~\ref{le:correspondence1}.
\end{proof}

\noindent
We are now ready to prove Theorem~\ref{th:apsp}. Consider any vertex $u$ in $K$. We first mark every vertex that can be reached from $(u,0)$ in $H$ with the weight of its shortest path from $(u,0)$. By Observation~\ref{obs:vertices-edges}, $H$ has $O(Bn)$ vertices and $O(Bn^2)$ edges, hence this can be done in $O(B n^2 + Bn \log (Bn))$ time~\cite{ft-fhuin-87}. For every $0\leq \beta \leq B$ and for every vertex $v\neq u$, by Corollary~\ref{cor:equivalence} the weight of a $\beta$-bounded cost shortest path in $K$ is the same as the weight of a shortest directed path from $(u,0)$ to $(v,\beta)$ in $H$. Hence, for every $0\leq \beta \leq B$ and for every vertex $v\neq u$, we can determine in total $O(B n^2 + Bn \log (Bn))$ time the weight of a $\beta$-bounded cost shortest path in $K$ connecting $u$ and $v$. Thus, for every $0\leq \beta \leq B$ and for every pair of vertices $u$ and $v$ in $K$, we can determine in total $O(B n^3 + Bn^2 \log (Bn))$ time the weight of a $\beta$-bounded cost shortest path in $K$ connecting $u$ and $v$. This concludes the proof of Theorem~\ref{th:apsp}.

\section{Arbitrary Costs and Weights}

Our algorithms, as many afore-mentioned approximation algorithms for the \bcmd problem, use a clustering approach as a first phase to find a set $C$ of $B+1$ cluster centers. The idea of the algorithm is to create a minimum height rooted tree $T=(U\subseteq V,D)$, so that $C \subseteq U$, by adding a set of edges of total cost at most $B$ to $G$. We will prove that such a tree approximates an optimal $B$-augmentation.

\subsection{Clustering} \label{ssec:clustering}

We start by defining the clustering approach used to generate the $B+1$ cluster centers. Whereas a costly binary search is used in \cite{dz-mdnus-10} to guess the radius of the clusters, we adapt the approach of \cite{bgp-ianar-12} to our more general setting.

For two vertices $u,v$, we denote by $\dist_G(u,v)$ the distance between $u$ and $v$ in $G$.
For a vertex $u$ and a set of vertices $S$, we denote by $\dist_G(u,S)$ the minimum distance between $u$ and any vertex from $S$ in $G$,
i.e., $\dist_G(u,S) = \min_{v\in S} \{\dist_G(u,v)\}$.
For a set of vertices $S$, we denote by $\dist_G(S)$ the minimum distance between any two distinct vertices from $S$ in $G$,
i.e., $\dist_G(S) = \min_{u\in S} \{\dist_G(u,S\setminus \{u\})\}$.

The clustering phase computes a set $C = \{c_1, \dots, c_{B+1}\}$ of $B+1$ cluster centers as follows. Vertex $c_1$ is an arbitrary vertex in $V$;
for $2\leq i \leq B+1$, vertex $c_i$ is chosen so that $\dist_G(c_i,\{c_1,\dots,c_{i-1}\})$ is maximized. Ties are broken arbitrarily.

\begin{lemma} \label{le:clustering}
 The clustering phase computes in $O(Bn^2)$ time a set $C\subseteq V$ of size $B+1$ such that $\dist_G(v,C) \le D^B_{opt}$ for every vertex $v\in V$.
\end{lemma}
\begin{proof}
First, note that the above described algorithm can easily be implemented in $O(Bn^2)$ time using $B$ iterations
of Dijkstra's algorithm with Fibonacci heaps \cite{ft-fhuin-87}.
Let $c_{B+2}$ denote a vertex maximizing $\dist_G(c_{B+2},C)$, and denote
 this distance by $R$. By definition, $\dist_G(v,C) \le R$ for every $v\in V$.
 To prove the lemma it remains to show that $R\le D^B_{opt}$. For the sake of contradiction, assume $D^B_{opt} < R$.
 Then, $C \cup \{c_{B+2}\}$ is a set of $B+2$ vertices with pairwise distance larger than $D^B_{opt}$ in $G$. We prove the following claim.

\begin{claimx}\label{cla:manydist}
Let $G'$ be a weighted graph and let $C'$ be a set of vertices in $G'$ such that $\dist_{G'}(C') > D$. Then, for every graph $G''$ obtained from $G'$ by adding a single non-edge of $G'$ with non-negative weight, there is a set $C'' \subset C'$ with $|C''|=|C'|-1$ and with $\dist_{G''}(C'') > D$.
\end{claimx}

\begin{proof}
Let $(u,v)$ denote the edge that is added to $G'$ to obtain $G''$.  For the sake of contradiction, assume that there is no vertex $w\in C'$ such that $\dist_{G''}(C' \setminus \{w\}) > D$. That is, every set $C'' \subset C'$ with $|C''|=|C'|-1$ contains two vertices whose distance is at most $D$. Then, there are four vertices $w_1,w_2,w_3,w_4\in C$ such that $\dist_{G''}(w_1,w_2)\leq D$ and $\dist_{G''}(w_3,w_4)\leq D$, or there are three vertices $w_1,w_2,w_3\in C$ such that $\dist_{G''}(w_1,w_2)\leq D$, $\dist_{G''}(w_1,w_3)\leq D$, and $\dist_{G''}(w_2,w_3)\leq D$.

In the first case, since $\dist_{G''}(w_1,w_2) < \dist_{G'}(w_1,w_2)$ and $\dist_{G''}(w_3,w_4) < \dist_{G'}(w_3,w_4)$, we have that $(u,v)$ is an edge of any shortest path $P_{1,2}$ from $w_1$ to $w_2$ and of any shortest path $P_{3,4}$ from $w_3$ to $w_4$. Assume, without loss of generality, that $u$ is encountered before $v$ when traversing $P_{1,2}$ starting at $w_1$ and when traversing $P_{3,4}$ starting at $w_3$ (otherwise swap $w_1$ and $w_2$ and/or $w_3$ and $w_4$). Therefore, we get (1A) $\dist_{G'}(w_1,u)+\dist_{G'}(v,w_2) \le D$, and (1B) $\dist_{G'}(w_3,u)+\dist_{G'}(v,w_4) \le D$. However, since $\dist_{G'}(C')>D$, we have (1C) $\dist_{G'}(w_1,u)+\dist_{G'}(u,w_3) > D$, and (1D) $\dist_{G'}(w_2,v)+\dist_{G'}(v,w_4)> D$. Denote $K := \dist_{G'}(w_1,u) + \dist_{G'}(v,w_2) + \dist_{G'}(w_3,u) + \dist_{G'}(v,w_4)$. Inequalities (1A) and (1B) give $K \le 2D$, while inequalities (1C) and (1D) give $K > 2D$, a contradiction.

In the second case, denote by $P_{1,2}$, $P_{1,3}$, and $P_{2,3}$ three paths in $G''$ with weight at most $D$ connecting $w_1$ and $w_2$, connecting $w_1$ and $w_3$, and connecting $w_2$ and $w_3$, respectively. Since $\dist_{G'}(\{w_1,w_2,w_3\})>D$, all these paths use edge $(u,v)$. Without loss of generality, assume $\dist_{G'}(w_1,u) \le \dist_{G'}(w_1,v)$. Hence, both $P_{1,2}$ and $P_{1,3}$ reach $u$ before $v$ when traversing such paths starting at $w_1$.  Without loss of generality, assume that $P_{2,3}$ reaches $u$ before $v$ when traversing such path starting at $w_2$ (otherwise, swap $w_2$ and $w_3$). Therefore, we get (2A) $\dist_{G'}(w_1,u)+\dist_{G'}(v,w_2) \le D$, and (2B) $\dist_{G'}(w_2,u)+\dist_{G'}(v,w_3) \le D$. However, since $\dist_{G'}(\{w_1,w_2,w_3\})>D$, we have (2C) $\dist_{G'}(w_2,v)+\dist_{G'}(v,w_3) > D$, and (2D) $\dist_{G'}(w_1,u)+\dist_{G'}(u,w_2)> D$. Denote $L := \dist_{G'}(w_1,u) + \dist_{G'}(v,w_2) + \dist_{G'}(w_2,u) + \dist_{G'}(v,w_3)$. Inequalities (2A) and (2B) give $L > 2D$, while inequalities (2C) and (2D) give $L \le 2D$, a contradiction.
This concludes the proof of the claim.
\end{proof}

 \noindent
 Now, since $C \cup \{c_{B+2}\}$ is a set of $B+2$ vertices with pairwise distance larger than $D^B_{opt}$ in $G$, by iteratively using the claim we have that in any $B$-augmentation $G_B$ of $G$, we have a set of $B+2-|F|\ge 2$ vertices with pairwise distance greater than $D^B_{opt}$, thus contradicting the definition of $D^B_{opt}$. This concludes the proof of the lemma.
\end{proof}

\subsection{A minimum height tree}

Let $C$ be a set of $B+1$ cluster centers such that the $B+1$ clusters with centers at $C=\{ c_0, \ldots , c_B\}$ and radius $D^B_{opt}$ cover the vertices of $G$. This set can be computed as described in the previous section.

\begin{definition} \label{def:SPT}
Let $G=(V,E)$ be a graph together with a weight function $w:[V]^2 \rightarrow \mathbb{N}$. Let $C\subseteq V$ and let $u$ be a vertex in $V$. A Shortest Path Tree of $G$, $C$, and $u$, denoted by \spt$(G,C,u)$, is a tree $T$ rooted at $u$, spanning $C$, whose vertices and edges belong to $V$ and $E$, respectively, and such that, for every vertex $c\in C$, it holds $d_T(u,c)=d_G(u,c)$.
\end{definition}

\noindent
The {\em height} of a weighted rooted tree $T$, which is denoted by $\hbar(T)$, is the maximum weight of a path from the root to a leaf.

\begin{definition}
Let $G=(V,E)$ be a graph together with a weight function $w:[V]^2 \rightarrow \mathbb{N}$ and a cost function $c:[V]^2 \rightarrow \mathbb{N}^*$. Let $C\subseteq V$, let $u$ be a vertex in $V$, and let $B\geq 0$ be an integer. A Minimum Height$_B$ SPT of $G$, $C$, and $u$, denoted by \mhspt{B}$(G,c,u)$, is a \spt$(G_B,C,u)$ of minimum height over all $B$-augmentations $G_B$ of~$G$.
\end{definition}

\noindent
Let $G_B$ be a $B$-augmentation of $G$ with diameter $D^B_{opt}$.

\begin{lemma} \label{le:height-optimal}
 The height of a \mhspt{B}$(G,C,u)$ is at most $D^B_{opt}$.
\end{lemma}
\begin{proof}
By definition, we have (A) $\hbar(\mbox{\mhspt{B}}(G,C,u)) \leq \hbar(\mbox{\spt}(G_B,C,u))$. Since $G_B$ is a $B$-augmentation of $G$ with diameter $D^B_{opt}$, we have (B) $\hbar(\mbox{\spt}(G_B,C,u)) \leq D^B_{opt}$. Inequalities (A) and (B) together prove the lemma.
\end{proof}

\noindent
We now present a relationship between the \bcmd problem and the problem of computing a \mhspt{B}$(G,C,u)$.

\begin{lemma} \label{lem:4-apx}
Let $G'_B$ be a $B$-augmentation of $G$ such that it holds $\hbar(\mbox{\spt}(G'_B,C,u)) = \hbar(\mbox{\mhspt{B}}(G,C,u))$, for any $u\in V$. Then, the diameter of $G'_B$ is at most $4 \cdot D^B_{opt}$.
\end{lemma}
\begin{figure}[tb]
  \begin{center}
     \includegraphics[width=5.5cm]{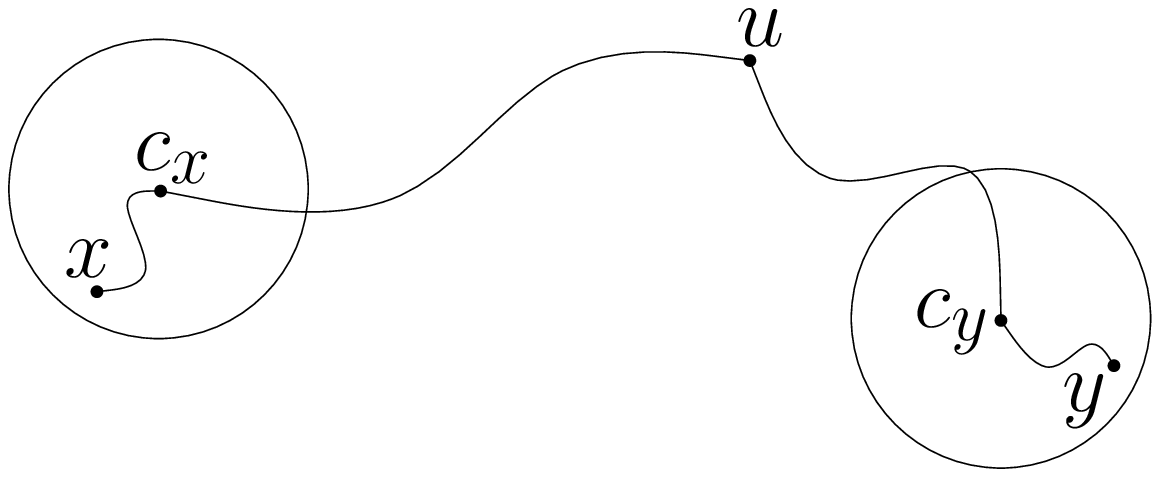}
     \caption{Illustrating the path defined in the proof of Lemma~\ref{lem:4-apx}.}\label{fig:ApproximationBound}
  \end{center}
\end{figure}
\begin{proof}
Consider two vertices $x$ and $y$ in $V$, see Figure~\ref{fig:ApproximationBound}. Let $c_x$ and $c_y$ be centers of the clusters $x$ and $y$ belong to, respectively. Then, we have  $\dist_{G'_B}(x,y) \leq \dist_{G}(x,c_x) + \dist_{G'_B}(c_x,u) +\dist_{G'_B}(u,c_y) +\dist_{G}(c_y,y)$. By Lemma~\ref{le:clustering}, $\dist_{G}(x,c_x),$ $\dist_{G}(c_y,y)\leq D_{opt}^B$. Since $\hbar(\mbox{\spt}(G'_B,C,u))=\hbar(\mbox{\mhspt{B}}(G,C,u))$ and by Lemma~\ref{le:height-optimal}, it holds $\dist_{G'_B}(c_x,u), \dist_{G'_B}(u,c_y)\leq D_{opt}^B$. Hence,  $\dist_{G'_B}(x,y) \leq 4 \cdot D_{opt}^B$.
\end{proof}

\subsection{Constructing a minimum height tree}
In this section, we show an algorithm to compute a \mhspt{B}$(G,C,c_1)$.

We introduce some notation and terminology. Let $C'=C\setminus\{c_1\}$. Observe that a \mhspt{B}$(G,C',c_1)$ is also a \mhspt{B}$(G,C,c_1)$, given that a \mhspt{B}$(G,C',c_1)$ contains $c_1$ as its root. Denote by $d_K^{j}(u,v)$ the minimum weight of a $j$-bounded cost path connecting $u$ and $v$ in $K$. For any $u\in V$, for any $S\subseteq C'$, and for any $0\leq j \leq B$, let $\gamma(u,S,j)$ denote the height of a \mhspt{j}$(G,S,u)$. Hence, the height of a \mhspt{B}$(G,C',c_1)$ is $\gamma(c_1,C',B)$. The following main lemma gives a dynamic programming recurrence for computing $\gamma(c_1,C',B)$.

\begin{lemma} \label{le:main}
For any $u\in V$, any $S\subseteq C'$, and any $0\leq j \leq B$, the following hold: If $|S|=1$, then $\gamma(u,S,j)=d_K^{j}(u,c_i)$ where $S=\{c_i\}$. If $|S|>1$, then $$\gamma(u,S,j)=\min_{\substack{v \in V \\ S' \subsetneq S \\j = j_1+j_2+j_3}} d_K^{j_1}(u,v)+\max\{\gamma(v,S',j_2), \gamma(v,S\setminus S',j_3)\}.$$
\end{lemma}
\begin{proof}
If $|S|=\{c_i\}$, then \mhspt{j}$(G,\{c_i\},u)$ is a minimum-weight path connecting $u$ and $c_i$ and having total cost at most $j$. Hence, $\gamma(u,S,j)=d_K^{j}(u,c_i)$. In particular, notice that, if $u=c_i$, then $\gamma(u,\{u\},j)=d_K^j(u,u)=0$.

If $|S|=m>1$, then suppose that the lemma holds for each $\gamma(u',S',j')$ with $|S'|\leq m-1$ by induction.
Denote by $T$ any \mhspt{j}$(G,S,u)$. Denote by $P(v,w)$ the unique path in $T$ connecting two vertices $v$ and $w$ of $T$.
We distinguish three cases, based on the structure of $T$.
In Case (a), the degree of $u$ in $T$ is at least two (see Figure~\ref{fig:dpr}(a)).
In Case (b), the degree of $u$ in $T$ is one and there exists a vertex $u'\in S$ such that every internal vertex of $P(u,u')$ has degree $2$ in $T$ and does not belong to $S$ (see Figure~\ref{fig:dpr}(b)).
Finally, in Case (c), the degree of $u$ in $T$ is one and there exists a vertex $u'\notin S$ such that every internal vertex of $P(u,u')$ has degree $2$ in $T$ and does not belong to $S$, and such that the degree of $u'$ is greater than two (see Figure~\ref{fig:dpr}(c)).

\begin{figure} [htb]
  \begin{center}
     \includegraphics[width=8cm]{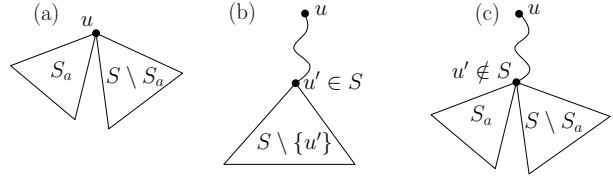}
     \caption{Illustration for the proof of Lemma~\ref{le:main}.} \label{fig:dpr}
  \end{center}
\end{figure}

First, we prove that one of the three cases always applies. If the degree of $u$ in $T$ is at least two, then Case (a) applies. Otherwise, the degree of $u$ is $1$. Traverse $T$ from $u$ until a vertex $v'$ is found such that $v'\in S$ or the degree of $v'$ is at least $3$. If $v'\in S$, then every internal vertex of $P(u,u')$ has degree $2$ in $T$ and does not belong to $S$, hence Case (b) applies. If $v'\notin S$, then the degree of $v'$ is at least $3$, and every internal vertex of $P(u,u')$ has degree $2$ in $T$ and does not belong to $S$, hence Case (c) applies. We now discuss the three cases.

In Case (a), $T$ is composed of two subtrees \mhspt{x}$(G,S_a,u)$ and \mhspt{y}$(G,S\setminus S_a,u)$, only sharing vertex $u$, with $\emptyset \subsetneq S_a \subsetneq S$. The height of $T$ is the maximum of the heights of \mhspt{x}$(G,S_a,u)$ and \mhspt{y}$(G,S\setminus S_a,u)$; also the cost of $T$ is at most $x+y$.
By induction, the heights of \mhspt{x}$(G,S_a,u)$ and \mhspt{y}$(G,S\setminus S_a,u)$ are $\gamma(u,S_a,x)$ and $\gamma(u,S\setminus S_a,y))$, respectively. Thus, the height of $T$ is $\max\{\gamma(v,S_a,x), \gamma(v,S\setminus S_a,y)\}$ and hence $\gamma(u,S,j) = \max\{\gamma(u,S_a,x), \gamma(u,S\setminus S_a,y)\}$. Such a value is found by the recursive definition of $\gamma(u,S,j)$ with $v=u$, $S'=S_a$, $j_1=0$, $j_2=x$, and $j_3=y$.

In Case (b), $T$ is composed of a path from $u$ to $u'$ with cost $x$ and weight $d_K^{x}(u,u')$, and of a \mhspt{y}$(G,S\setminus \{u'\},u')$. The height of $T$ is the sum of $d_K^{x}(u,u')$ and the height of \mhspt{y}$(G,S\setminus \{u'\},u')$; also the cost of $T$ is at most $x+y$.
By induction, the height of \mhspt{y}$(G,S\setminus \{u'\},u')$ is $\gamma(u',S\setminus \{u'\},y)$. Thus, the height of $T$ is $d_K^{x}(u,u') + \gamma(u',S\setminus \{u'\},y)$ and hence $\gamma(u,S,j)= d_K^{x}(u,u') + \gamma(u',S\setminus \{u'\},y)$. Such a value is found by the recursive definition of $\gamma(u,S,j)$ with $v=u'$, $S'=S\setminus \{u'\}$, $j_1=x$, $j_2=y$, and $j_3=0$.

In Case (c), $T$ is composed of a path from $u$ to $u'$ with cost $x$ and weight $d_K^{x}(u,u')$, of a \mhspt{y}$(G,S_a,u')$, and of a \mhspt{z}$(G,S\setminus S_a,u')$ with $\emptyset \subsetneq S_a \subsetneq S$.
The height of $T$ is the sum of $d_K^{x}(u,u')$ and the maximum between the heights of \mhspt{y}$(G,S_a,u')$ and \mhspt{z}$(G,S\setminus S_a,u')$; also the cost of $T$ is at most $x+y+z$.
By induction, the heights of \mhspt{y}$(G,S_a,u')$ and \mhspt{z}$(G,S\setminus S_a,u')$ are $\gamma(u',S_a,y)$ and $\gamma(u',S\setminus S_a,z)$, respectively. Thus, the height of $T$ is $d_K^{x}(u,u') + \max\{\gamma(u',S_a,y),$ $\gamma(u',S\setminus S_a,z)\}$ and hence $\gamma(u,S,j)= d_K^{x}(u,u') + \max\{\gamma(u',S_a,y),\gamma(u',S\setminus S_a,z)\}$. Such a value is found by the recursive definition of $\gamma(u,S,j)$ with $v=u'$, $S'=S_a$, $j_1=x$, $j_2=y$, and $j_3=z$.

This concludes the induction and hence the proof of the lemma.
\end{proof}

\noindent
Lemma~\ref{le:main} yields the following.

\begin{theorem} \label{thm:main}
There exists a $(1,4)$-approximation algorithm for the \bcmd problem with $O((3^B B^3 + n + \log (Bn)) B n^2)$ running time.
\end{theorem}

\begin{proof}
Given an instance $(G,w,c,B)$ of the \bcmd problem, by Theorem~\ref{th:apsp} we can determine, for every pair of vertices $u,v\in V$ and for every $1\leq j\leq B$, the minimum weight of a $j$-bounded cost path connecting $u$ and $v$ in total $O((n + \log (Bn)) B n^2)$ time.
By Lemma~\ref{le:clustering}, a clustering of $G$ can be computed in $O(Bn^2)$ time. Due to Lemma~\ref{le:main}, the problem of computing a \mhspt{B}$(G,C\setminus\{c_1\},c_1)$ can be solved by dynamic programming over the triples $(u,S,j)$ (there are $O\left(B \binom{B+1}{s} n\right)$ such triples with $|S|=s$); the computation of the value for such a triple requires to take a minimum over $j^3 2^{|S|} n$ values, hence the dynamic programming running time is $O\left(n B \sum_{s=0}^{B} \binom{B+1}{s} B^3 2^{s} n\right) = O(B^4 3^B n^2)$. Observe that the dynamic programming can be designed in such a way that a rooted tree with height equal to $\gamma(u,S,j)$ is computed together with the value of $\gamma(u,S,j)$. This is trivially done in the base case; moreover, in the inductive case it only requires, for each $v \in V$, each $S' \subsetneq S$, and each $j = j_1+j_2+j_3$, the computation of a shortest path tree. Finally, by Lemma~\ref{lem:4-apx}, augmenting $G$ with the non-edges that are present in a \mhspt{B}$(G,C\setminus \{c_1\},c_1)$ yields a $B$-augmentation $G_B$ whose diameter is at most $4 \cdot D^B_{opt}$.
\end{proof}

\section{Unit Costs and Arbitrary Weights}

For the special case in which each edge has unit cost and arbitrary weight, our techniques lead to several results, that are described in the following. Observe that, in this case we are allowed to insert in $G$ exactly $k$ non-edges of $G$, where $k=B=O(n^2)$. We remark that Theorem~\ref{thm:main} gives a $(1,4)$-approximation algorithm running in
$O((3^k k^3 + n) k n^2)$ time for this special case.

In the following, we denote by $C$ a clustering with $k+1$ clusters constructed as described in Subsection~\ref{ssec:clustering}. We first show a $((k+1)^2,3)$-approximation algorithm.

\begin{theorem}
Given an instance of the \bcmd problem with unit costs, there exists a $((k+1)^2,3)$-approximation algorithm with 
$O(k n^3)$ running time.
\end{theorem}
\begin{proof}
For every pair of cluster centers $c_i,c_j \in C$ compute a shortest path in $K$ between $c_i$ and $c_j$ that contains at most $k$ non-edges of $G$. Add those edges to $F$ and let $G'=(V,E \cup F)$. By Theorem~\ref{th:apsp} and since $k=O(n^2)$, $G'$ can be constructed in $O(k n^3)$ time. Observe that, for each pair of cluster centers, the algorithm adds at most $k$ non-edges of $G$ to $F$, thus at most $k(k+1)^2$ non-edges in total. We prove that, for every $v_i,v_j\in V$, there exists a path in $G'$ connecting $v_i$ and $v_j$ whose weight is at most $3\cdot D_{opt}^k$. Denote by $c_i$ and $c_j$ the centers of the clusters $v_i$ and $v_j$ belong to, respectively. We have $\dist_{G'}(v_i,v_j) \leq \dist_{G'}(v_i,c_i) + \dist_{G'}(c_i,c_j) + \dist_{G'}(c_j,v_j)$. By Lemma~\ref{le:clustering}, $\dist_{G'}(v_i,c_i),\dist_{G'}(v_j,c_j)\leq D_{opt}^k$; also, by construction, $\dist_{G'}(c_i,c_j)\leq D_{opt}^k$, and the theorem follows.
\end{proof}

\noindent
Next, we give a $(k,4)$-approximation algorithm.

\begin{theorem}
Given an instance of the \bcmd problem with unit costs, there exists a $(k,4)$-approximation algorithm with 
$O(k n^2)$ running time.
\end{theorem}
\begin{proof}
Pick an arbitrary cluster center, say $c_1$. For every cluster center $c_j \in C \setminus \{c_1\}$, compute a shortest path between $c_1$ and $c_j$ in $K$  containing at most $k$ non-edges of $G$. Add those edges to $F$ and let $G'=(V,E \cup F)$. By Corollary~\ref{cor:equivalence}, a shortest path between $c_1$ and $c_j$ in $K$ containing at most $k$ non-edges of $G$ corresponds to a shortest path between $(c_1,0)$ and $(c_j,k)$ in digraph $H$. By Observation~\ref{obs:vertices-edges}, $H$ has $O(kn)$ vertices and $O(kn^2)$ edges. Hence, Dijkstra's algorithm with Fibonacci heaps \cite{ft-fhuin-87} computes all the shortest paths between $(c_1,0)$ and $(c_j,k)$, for every $c_j \in C \setminus \{c_1\}$, in total $O(kn^2)$ time. Observe that, for each cluster different from $c_1$, the algorithm adds at most $k$ non-edges of $G$ to $F$, thus at most $k^2$ non-edges in total. We prove that, for every $v_i,v_j\in V$, there exists a path in $G'$ connecting $v_i$ and $v_j$ whose weight is at most $4\cdot D_{opt}^k$. Denote by $c_i$ and $c_j$ the centers of the clusters $v_i$ and $v_j$ belong to, respectively. We have $\dist_{G'}(v_i,v_j) \leq \dist_{G'}(v_i,c_i) + \dist_{G'}(c_i,c_1) + \dist_{G'}(c_1,c_j) + \dist_{G'}(c_j,v_j)$. By Lemma~\ref{le:clustering}, $\dist_{G'}(v_i,c_i),\dist_{G'}(v_j,c_j)\leq D_{opt}^k$; by construction, $\dist_{G'}(c_i,c_1),\dist_{G'}(c_1,c_j)\leq D_{opt}^k$, and the theorem follows.
\end{proof}

\noindent
Finally, we present a $(1,3k+2)$-approximation algorithm.

\begin{theorem}
Given an instance of the \bcmd problem with unit costs, there exists a $(1,3k+2)$-approximation algorithm with $O(n^2 + k^2)$ running time.
\end{theorem}
\begin{proof}
For every pair of clusters $C_i$ and $C_j$, with $1\leq i <j \leq k+1$, let $e_{ij}$ be the edge of minimum weight connecting a vertex in $C_i$ with a vertex in $C_j$. We denote by $F'$ the set of these edges. For a subset $F$ of $F'$, we say that $F$ {\em spans} $C$ if the graph representing the adjacencies between clusters via the edges of $F$ is connected. Let $F$ be a minimum-weight set of $k$ edges from $F'$ spanning $C$. Let $G'=(V,E \cup F)$. The set $F'$, and hence the graph $G'$, can be constructed in $O(n^2 + k^2)$ time as follows.
%
%
Consider all the edges of $K$ and keep, for each pair of clusters, the edge with smallest weight. This can be done in $O(n^2)$ time. Finally, compute in $O(k^2)$ time a minimum spanning tree of the resulting graph
\cite{FredmanW94}, that has $O(k)$ vertices and $O(k^2)$ edges. Observe that the algorithm adds at most $k$ non-edges of $G$ to $F$.
We prove that, for every $v_i,v_j\in V$, there exists a path in $G'$ connecting $v_i$ and $v_j$ whose weight is at most $(3k+2) D_{opt}^k$. Denote by $P_C$ the (unique) subset of $F$ connecting the clusters $v_i$ and $v_j$ belong to. Let $(x_1,y_1), (x_2,y_2), \ldots, (x_m,y_m)$ be the edges of $P_C$ in order from $v_i$ to $v_j$. Then, $\dist_{G'}(v_i,v_j) \leq \dist_{G}(v_i,x_1) + w(x_1,y_1) + \dist_{G}(y_1,x_2) + \ldots + w(x_m,y_m) + \dist_{G}(y_m,v_j)$. By Lemma~\ref{le:clustering}, $\dist_{G}(y_i,x_{i+1})\leq 2 D_{opt}^k$, and $\dist_{G}(v_i,x_1), \dist_{G}(y_m,v_j)\leq D_{opt}^k$. Also, $w((x_i,y_i))\leq D_{opt}^k$, and the theorem follows.
\end{proof}



\section{Hardness Results}

\newcommand{\ubcmd}{\textsc{u-bcmd}\xspace}

The main theorem of this section provides a parameterized intractability result for \bcmd{} with unit weights and unit costs, and some related problems.
The \ubcmd problem has as input an unweighted graph $G=(V,E)$ and two integers $k$ and $d$, and the question is whether there is a set $F\subseteq [V]^2 \setminus E$, with $|F|\le k$, such that the graph $(V,E\cup F)$ has diameter at most $d$.
The parameter is $k$.
We will show that \ubcmd is $\W[2]$-hard. We will also provide refinements to the minimum conditions required for intractability, namely \ubcmd remains \NP-complete for graphs of diameter $3$ with target diameter $d=2$. We note that although Dodis and Kanna~\cite{dk-dnbpd-99} provide an inapproximability reduction from \textsc{Set Cover}, they begin with a disconnected graph, and expand the instance with a series of size $2$ sets, which does not preserve the size of the optimal solution, and therefore their reduction cannot be used to show parameterized complexity lower bounds.

%

\begin{theorem}\label{thm:reduction}
\textsc{Set Cover} is polynomial-time reducible to \ubcmd{}. Moreover, the reduction is parameter preserving and creates an instance with diameter $3$ and target diameter $2$.
\end{theorem}
\begin{proof}
Let $(X,S,k)$ be an instance of \textsc{Set Cover} where $S$ is the base set and $X \subset \mathcal{P}(S)$ is the set from which we must pick the set cover of $S$ with size at most $k$. We construct an instance $(G=(V,E),k,d)$ of \ubcmd{} as follows.

Let $m = \Card{X}\cdot k$. The vertex set $V$ is the disjoint union of $5$ sets:
\begin{itemize}
\item a set $Y$ corresponding to the set $X$ where for each $x \in X$ we have a vertex $y \in Y$,
\item a set $T = \biguplus_{i \in [m]} T_{i}$ corresponding to $S$ where, for each $s \in S$ and $i \in [m]$, we have a vertex $t_{i} \in T_{i}$ (\emph{i.e.,} we have $m$ copies of a set of vertices corresponding to~$S$),
\item a set $U$ with $\binom{m}{2}$ vertices $u_{ij}$, one for each pair of disjoint subsets $T_{i}$, $T_{j}$ of $T$ (where $i\neq j$),
\item the set $\{a\}$, and
\item the set $\{b\}$.
\end{itemize}
The edge set $E$ consists of the following edges:
\begin{itemize}
\item $ab$,
\item $by$ for each vertex $y \in Y$,
\item $bu_{ij}$ for each vertex $u_{ij} \in U$,
\item $yy'$ for each pair of vertices $y,y' \in Y$,
\item $yt_{i}$ for each pair of vertices $y\in Y$ and $t_{i} \in T_{i}$ for each $i \in [m]$ where the corresponding element $s\in S$ is in the corresponding set $x \in X$ in the \textsc{Set Cover} instance,
\item $t_{i}u_{jl}$ for each pair of vertices $t_{i}\in T_{i}$ and $u_{jl} \in U$ such that $i\in \{j,l\}$, and
\item $u_{ij}u_{lp}$ for each pair of vertices $u_{ij}, u_{lp} \in U$.
\end{itemize}

\begin{figure}
\begin{center}
\includegraphics[scale=0.7]{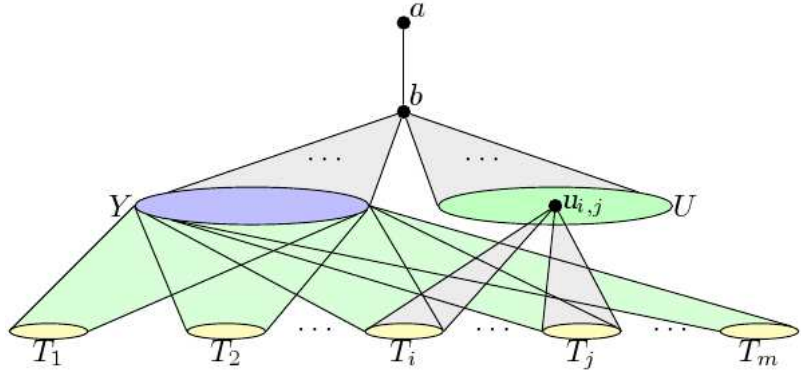}
\end{center}
\caption{Sketch of the construction for the \textsc{Set Cover} to \ubcmd{} reduction. The edge sets represented in gray are complete, the edge sets represented in light green correspond to the set membership from the \textsc{Set Cover} instance. The vertex sets $Y$ and $U$ are cliques. The vertex sets $T_{i}$ are independent sets for all $i \in [m]$.}\label{fig:construction}
\end{figure}

\noindent
We set $d = 2$. Note that $k$ in the \ubcmd{} instance is the same $k$ as for the \textsc{Set Cover} instance.
The construction is sketched in Figure~\ref{fig:construction}.

\begin{claimx}\label{cla:dist2}
For all $v,v' \in V\setminus\{a\}$ we have $\dist(v,v') \leq 2$.
\end{claimx}

\begin{proof}
The vertices of $U$ are at distance one from each other. The vertices of $Y$ are at distance one from each other.

If one of the vertices is $b$, clearly $b$ is at distance $1$ from the vertices of $Y$ and $U$. Therefore the vertices of $U$ and $Y$ are at most distance $2$ from each other via the path through $b$.

Each vertex $t \in T$ is at distance one from some vertex $y \in Y$. As $Y$ is a clique, $t$ is at distance at most two from all the vertices in $Y$.

Each vertex $t \in T$ is at distance one from some vertex $u \in U$. As $U$ is also a clique, $t$ is at distance at most two from all vertices of $U$.

For each pair of vertices $t_{i} \in T_{i}$ and $t_{j} \in T_{j}$ there is a vertex $u_{ij} \in U$ such that $t_{i}u_{ij} \in E$ and $t_{j}u_{ij} \in E$. If $i=j$ then any vertex $u_{ik} \in U$ will suffice. Thus all the vertices of $T$ are at most distance 2 from each other.
\end{proof}

\begin{claimx}\label{cla:dist3}
For all $v \in V$ we have $\dist(a,v) \leq 3$. Moreover, $\dist(a,v) = 3$ if and only if $v \in T$.
\end{claimx}

\begin{proof}
As the distance from $b$ to all other vertices is at most $2$, the distance from $a$ to all other vertices is at most $3$. Moreover, as the distance from $b$ to the vertices of $U$ and $Y$ is one, the distance from $a$ to these vertices is two. Therefore the only vertices at distance three from $a$ are the vertices of $T$.
\end{proof}

\noindent
Thus we are concerned only with reducing the distance between $a$ and the vertices of $T$.

\begin{claimx}\label{cla:corr-red}
$(X,S,k)$ is a \Yes{}-instance of \textsc{Set Cover} if and only if $(G,k,d)$ is a \Yes{}-instance of \ubcmd{}.
\end{claimx}

\begin{proof}
Let $X' \subseteq X$ be the set cover that witnesses that $(X,S,k)$ is a \Yes{}-instance of \textsc{Set Cover}. Let $Y'\subseteq Y$ be the set of vertices that corresponds to $X'$. We have $\Card{Y'} = \Card{X'} \leq k$. If we add the edges $ay$ for all $y \in Y'$, then $a$ is at distance at most $2$ from all vertices $t \in T$. As $X'$ is a set cover of $S$, for each $s\in S$ there is at least one set $x \in X'$ such that $s \in x$. Then there is an edge from $a$ to the vertex $y$ corresponding to $x$, and by the construction, $y$ is adjacent to $t \in T$ if and only if the corresponding element $s$ is in $S$, thus we have a path $a\rightsquigarrow y \rightsquigarrow t$.

Now, assume $(G,k,d)$ is a \Yes{}-instance of \ubcmd{}. First consider the case where all the edges are added between $a$ and the vertices of $Y$. Then the set $Y' \subseteq Y$ of vertices newly adjacent to $a$ corresponds to a set cover $X' \subseteq X$ in the same way as before.

We must demonstrate that we may only (productively) add edges between $a$ and $Y$. Clearly we cannot add the edge $ab$, as it already exists, and clearly adding edges from $b$ to other vertices is not necessary, as we could simply add the edges directly to $a$. Suppose we add edges between $a$ and $T$ directly, as there are $\Card{S}\cdot m$ vertices in $T$, we clearly cannot reduce the distance between $a$ and all the vertices in $T$, so some edges must still be added elsewhere. If we add edges between $a$ and $U$ we can reduce the distance between $a$ and at most $2k$ vertices of $T$. Thus even were we to add such edges, there is at least one $T_{i} \subset T$ still at distance $3$ from $a$. Therefore we must add edges from $a$ to $Y$ such that $T_{i}$ is dominated by this subset of $Y$. Clearly this corresponds to a set cover of $S$.
\end{proof}

\noindent
We note that the reduction is obviously polynomial-time computable, and the parameter $k$ is preserved.
The theorem now follows from the previous claims.
\end{proof}

\begin{corollary}\label{cor:np-complete}
\ubcmd{} is \NP{}-complete even for graphs of diameter three with target diameter two.
\end{corollary}

\begin{proof}
As it is already known that \ubcmd{} is in \NP{}~\cite{dk-dnbpd-99}, the result for \ubcmd{} follows from Theorem~\ref{thm:reduction}.
\end{proof}

\noindent
As \textsc{Set Cover} is $\W[2]$-hard with parameter $k$, combined with Corollary~\ref{cor:np-complete} we also have the following result.

\begin{corollary}\label{cor:w[2]-hard-1}
\ubcmd{} is $\W[2]$-hard even for graphs of diameter three with target diameter two.
\end{corollary}

%
%

\noindent
We note additionally that as the initial graph has diameter $3$ and the target diameter is $2$, it is even \NP{}-hard and \W[2]-hard to decide if there is a set of $k$ new edges that improves the diameter by one. Furthermore by taking $a$ as source vertex, the results transfer immediately to the single-source version as discussed by Demaine \& Zadimoghaddam~\cite{dz-mdnus-10}.

The construction of Theorem~\ref{thm:reduction} can even be extended to give a parameterized inapproximability result for \ubcmd{}.

\begin{theorem}\label{thm:param_inapprox}
It is $\W[2]$-hard to compute a $(1+\frac{c}{k}, \frac{3}{2}-\varepsilon)$-approximation for \ubcmd{} for any constants $c$ and $\varepsilon > 0$.
\end{theorem}
\begin{proof}
We repeat the construction of Theorem~\ref{thm:reduction}, except that we introduce $c+1$ copies of the $Y$ and $T$ components and set $k' = k\cdot(c+1)$. Let $Y_{i}$ with $1\leq i\leq c+1$ be the copies of the $Y$ components and $T_{i,j}$ with $1\leq i \leq c+1$ and $1 \leq j \leq m$ be the copies of the $T$ components. The edges are similar to the previous construction; we highlight the differing edges:
\begin{itemize}
\item $yy'$ for each $y,y' \in \bigcup_{i}Y_{i}$,
\item $y_{i}t_{i,j}$ for each $y_{i} \in Y_{i}$ and $t_{i,j} \in T_{i,j}$ where the corresponding element $x\in X$ is in the corresponding set $S$,
\item $by$ for all $y \in \bigcup_{i}Y_{i}$, and
\item $t_{a,i}u_{i,j}$ for each vertex $t_{a,i} \in \bigcup_{a}T_{a,i}$ and each vertex $u_{i,j}\in U$.
\end{itemize}

\noindent
Then apart from $a$, all vertices remain at pairwise distance $2$, with $a$ at distance $3$ from vertices in $\bigcup_{i}Y_{i}$. Clearly to reduce the diameter to $2$ we require the addition of edges from $a$ to vertices of the $Y$ component copies as before, furthermore we require edges to each copy, otherwise there is some $T_{i,j}$ component that remains at distance $3$ from $a$. Thus if the \textsc{Set Cover} instance has a solution of size $k$, the \ubcmd{} instance has a solution of size $k'$.

Let $F$ be a set of $(1+\frac{c}{k'})k' = k'+c$ edges such that the diameter of $G'=(V,E \cup F)$ is at most $(\frac{3}{2}-\varepsilon)\cdot 2$. Since the diameter of $G'$ is integral, it has diameter at most 2. Since there are
$c+1$ copies of $Y$, at least one of them, $Y_i$, has at most $k'$ vertices adjacent to $a$, giving a size $k'$ set cover as before.
\end{proof}

\subsection*{Acknowledgments}
SG acknowledges support from the Australian Research Council
(grant DE120101761).
JG is funded by the Australian Research Council FT100100755.
NICTA is funded by the Australian
Government as represented by the Department of Broadband,
Communications and the Digital Economy and the Australian Research
Council through the ICT Centre of Excellence program.

\bibliographystyle{plain}
\bibliography{references}
\end{document}